\documentclass[letterpaper, 10 pt, conference]{ieeeconf}  % Comment this line out
\IEEEoverridecommandlockouts                              % This command is only
\usepackage{tikz}
\usepackage{verbatim}
\usetikzlibrary{arrows}

\usetikzlibrary{%
    decorations.pathreplacing,%
    decorations.pathmorphing%
}

\usetikzlibrary{calc,patterns,
                 decorations.pathmorphing,
                 decorations.markings}

\usetikzlibrary{decorations.pathmorphing} % for snake lines
\usetikzlibrary{matrix} % for block alignment
\usetikzlibrary{arrows} % for arrow heads
\usetikzlibrary{calc} % for manimulation of coordinates

\tikzstyle{block} = [draw,rectangle,thick,minimum height=2em,minimum width=2em]
\tikzstyle{sum} = [draw,circle,inner sep=0mm,minimum size=2mm]
\tikzstyle{connector} = [->,thick]
\tikzstyle{line} = [thick]
\tikzstyle{branch} = [circle,inner sep=0pt,minimum size=1mm,fill=black,draw=black]
\tikzstyle{guide} = []
\tikzstyle{snakeline} = [connector, decorate, decoration={pre length=0.2cm,
                         post length=0.2cm, snake, amplitude=.4mm,
                         segment length=2mm},thick, magenta, ->]

\usepackage{amsmath,amssymb}
\usepackage{graphicx}
\usepackage{color}
\usepackage{indentfirst} %%%Ê׶οոñ

\usepackage{booktabs}%%%excel biaoge ºê
\usepackage{multirow}

\usepackage{subfigure}

\newtheorem{assumption}{Assumption}
\newtheorem{theorem}{Theorem}

\newtheorem{lemma}{Lemma}

\newtheorem{problem}{Problem}

\renewcommand{\baselinestretch}{1}

\title{\LARGE \bf
Guaranteed Performance Leader-follower Control for Multi-agent Systems with
Linear IQC-Constrained Coupling\thanks{This work was supported by the Australian Research Council under projects DP0987369 and DP120102152.}
}

\author{Yi Cheng and and Valery A. Ugrinovskii% <-this % stops a space
\thanks{The work of V. Ugrinovskii was carried in part while he was a visitor at the Australian National University.}% <-this % stops a space
\thanks{Yi Cheng and Valery A. Ugrinovskii are with the School of Engineering
  and Information Technology, The University of New South Wales at
  the Australian Defence Force Academy, Canberra, ACT 2600, Australia.
      Email: {\tt\small yi.cheng@student.adfa.edu.au, v.ougrinovski@adfa.edu.au}}%
}

\begin{document}

\maketitle
\thispagestyle{empty}
\pagestyle{empty}

%%%%%%%%%%%%%%%%%%%%%%%%%%%%%%%%%%%%%%%%%%%%%%%%%%%%%%%%%%%%%%%%%%%%%%%%%%%%%%%%
\begin{abstract}

This paper considers the leader-follower control problem for a linear
multi-agent system with undirected topology and linear coupling
subject to integral quadratic constraints (IQCs). A consensus-type control
protocol is proposed based on each agent's states relative its
neighbors. In addition a selected set of agents uses for control their states
relative the leader. Using a coordinate transformation, the consensus
analysis of the multi-agent system is recast as a decentralized robust
control problem for an auxiliary interconnected large scale system. Based on
this interconnected large scale system, sufficient conditions are obtained
which guarantee that the system tracks the leader. These conditions
guarantee a suboptimal bound on the system tracking performance. The
effectiveness of the proposed method is demonstrated using a simulation
example.

\end{abstract}

%%%%%%%%%%%%%%%%%%%%%%%%%%%%%%%%%%%%%%%%%%%%%%%%%%%%%%%%%%%%%%%%%%%%%%%%%%%%%%%%
\section{Introduction}
% no \IEEEPARstart

Theoretical study of distributed coordination and control for multi-agent
systems has received increasing attention in the past decade, due to its
broad applications in unmanned air vehicles (UAVs), formation control,
flocking, distributed sensor networks, etc \cite{[Olfati2004]}. As a result, much progress has been made in the study of cooperative control of multi-agent systems \cite{[Olfati2007]}, \cite{[Ren2007]}, \cite{[A. Arenas]}.

Efforts have recently been made to consider the leader-following consensus
problem. %The leader-following problem is to enable a group of agents to
%track a leader, which is independent of all other agents.
It was noted
in \cite{[Meng2011]} that the control problem becomes much more complex if
only a portion of the agents in the group has access to the leader. %In this case, ideas based on consensus play an instrumental role in reaching a
%solution.
For example, the leader-following consensus problem for higher
order multi-agent systems is presented for both fixed and switching
topologies in \cite{[Ni2010]}. In \cite{[Hong2008]}, distributed observers
are designed for the system of second-order agents where an active leader
to be followed moves with an unknown velocity, and the interaction topology
has a switching nature. The consensus-based approach to observer-based
synchronization of multiagent systems to the leader has further been explored
in~\cite{U6,U8}.

The majority of leader following consensus-based control problems currently
considered in the literature are focused on systems of independent agents. %where each individual agent has no influence on other agents except via
%control actions.
In many physical systems, however, interactions between
agents are inevitable and must be taken into account. Examples of
systems with a dynamical interaction between subsystems include power
systems and spacecraft control systems \cite{[Siljak]}. While in
interacting systems decentralized control schemes can be used, in the
examples mentioned above only relative states are available for measurement
which poses an additional difficulty when using decentralized
control. Hence, consensus based control strategies using relative state
information have application in these problems; e.g., see \cite{[Chen2010]}.

In this paper, we are concerned with the leader-follower control problem
for multi-agent systems coupled via linear unmodelled dynamics. Compared
with the existing work in the field of the leader-follower consensus
problem, interactions between the agents are regarded as uncertainty and are
described in terms of time-domain integral quadratic constraints (IQCs)
\cite{[Ian2000]}. The IQC modelling is a well established technique to
describe uncertain interactions between subsystems in a large scale system
\cite{[Ugrinovskii2005]}, \cite{[Li2007]}, \cite{[Ugrinovskii2000]}.
%The dynamical linear coupling model described in terms of IQCs, which is proposed
%in this paper, captures a large class of interconnections including
%unmodelled linear dynamics, uncertain input delays and norm-bound
%uncertainties, as special cases.

The IQC modelling allows us to analyze the effects of interactions between
the agents from a robustness viewpoint. In this respect, the recent paper
\cite{[Trentelman2013]} is worth mentioning, which considers robust consensus
protocols for synchronization of multiagent systems under additive
uncertain perturbations with bounded $H_\infty$ norm. Since the IQC conditions
in our paper capture uncertain perturbations with bounded $L_2$ gain, we
note a similarity between the two uncertainty classes. However, thanks to
the time-domain IQC modelling, our paper goes beyond establishing a
robust consensus. It develops the optimization approach to the
leader-follower tracking problem, which provides a guarantee of performance
of the leader-follower system under consideration (note that
\cite{[Trentelman2013]} consider a leaderless network). Namely, we pose the
leader-follower tracking problem as a constrained optimization problem in
which we optimize the worst-case consensus tracking performance of the
system, as well as the cost associated with protocol actions. Our interest
in robust performance guarantees is inspired by recent results on the
distributed LQR design~\cite{[Borrelli2008],[Hongwei2011]}.

%This leads to a
%method for designing a distributed controller for
%the system of coupled agents, where local tuning parameters can be
%chosen to minimize the bound on the consensus performance of the protocol
%leading to a guaranteed performance. The method is based on the reduction
%of the original problem to an optimization problem involving parameterized
%linear matrix inequalities (LMIs). Furthermore, we show that the
%design of the tracking protocol can be simplified using decoupled
%LMIs. This however leads to a weaker tracking result in that we can only
%guarantee a greater bound on consensus tracking performance.

The main contribution of the paper is a sufficient condition for the
design of a guaranteed consensus tracking performance protocol for
multi-agent systems subject to uncertain linear coupling. To derive such a
condition, the underlying guaranteed performance leader following control
problem is transformed into a guaranteed cost decentralized robust control
for an auxiliary large scale system, which is comprised of coupled
subsystems. The fact that the subsystems remain
coupled after the transformation constitutes the main difference of our
approach, compared with, e.g.,  \cite{[Zhongkui2010],[Hongwei2011]}, where
similar transformations resulted in a set of completely decoupled stabilization
problems. Coupling between subsystems poses an additional difficulty
compared with \cite{[Zhongkui2010],[Hongwei2011]},  which stems from
coupling between the agents. It is overcome using the minimax control
design methodology of decentralized control synthesis
\cite{[Ugrinovskii2005],[Li2007],[Ugrinovskii2000]}.

The paper is organized as follows. Section \uppercase\expandafter{\romannumeral2} includes the problem formulation and some preliminaries. The main results are given in Section \uppercase\expandafter{\romannumeral3}. In section \uppercase\expandafter{\romannumeral4}, the illustrative example is presented. Finally, the conclusions are given in Section \uppercase\expandafter{\romannumeral5}.

\section{Problem Formulation and Preliminaries}
\label{problem formulation}

\subsection{Graph theory}
Consider an undirected graph $\mathcal {G}=(\mathcal {V}, \mathcal {E},
\mathcal {A})$, where $\mathcal {V}= \{\mathcal {V}_1, \mathcal
{V}_2,\cdots, \mathcal {V}_N\}$ is a finite nonempty node set and $\mathcal
{E} \subseteq \mathcal {V}\times \mathcal {V}$ is an edge set of unordered
pairs of nodes. The edge $(i,j)$ in the edge set of an undirected graph
means that nodes $i$ and $j$ obtain information from each
other. Node $i$ is called a neighbor of node $j$ if $(i,j)\in \mathcal
{E}$. The set of neighbors of node $i$ is defined as $N_i=\{j|(i,j)\in
\mathcal {E}\}$. $\mathcal {G}$ is a simple graph if it has no self-loops
or repeated edges. If there is a path between any two nodes of the graph
$\mathcal {G}$, then the graph $\mathcal {G}$ is connected, otherwise it is
disconnected. The adjacency matrix $\mathcal {A}=[a_{ij}]\in R^{N\times N}$
of the undirected graph $\mathcal {G}$ is defined as $a_{ij}=a_{ji}=1$ if
$(i,j)\in \mathcal {E}$, and $a_{ij}=a_{ji}=0$ otherwise. The degree matrix
$\mathcal {D} =\mathrm{diag}\{d_1,\cdots,d_N\}\in R^{N\times N}$ is a
diagonal matrix, whose diagonal elements are $d_i =
\sum\limits_{j=1}^{N}a_{ij}$ for $i = 1,\cdots,N$. The Laplacian matrix of
the graph is defined as
$
\mathcal {L} =\mathcal {D} - \mathcal {A}.
$
It is symmetric when $\mathcal {G}$ is undirected.

\subsection{Problem Formulation}
 Consider a system consisting of $N$ agents and a leader. The communication
 topology between the agents is described by a simple undirected connected graph $\mathcal {G}$. % All $N$ agents are assumed to be identical linear dynamical agents, coupled with their neighbors via linear coupling.
 Dynamics of the $i$th agent are described by the equation
\begin{equation} \label{agents dymamic}
 \dot{x}_i=Ax_i+B_1u_i+ B_2\sum\limits_{j\in N_i} \varphi (x_j(.) \mid_0^t-x_i(.) \mid_0^t),
\end{equation}
where the notation $\varphi(y (.) \mid_0^t)$ describes an operator mapping
functions $y(s)$, $0\leq s \leq t$, into $\Re^n$. Also, $x_i\in \Re^n$ is
the state, $u_i\in \Re^p$ is the control input. We note that the last term
in (\ref{agents dymamic}) reflects a relative nature of interactions
between agents.

Let $L_{2e} [0, \infty)$ be the space of functions $y(.): [0, \infty)
\rightarrow \Re ^n$ such that $\int_{0}^{T}\|y(t)\|^2 dt < \infty,~\forall
T >0$.

\begin{assumption}\label{A1}
$\varphi$ is a linear operator mapping $L_{2e} [0,
\infty)\rightarrow L_{2e} [0, \infty)$ which satisfies the following
integral quadratic constraint (IQC) condition  \cite{[Ian2000]}. There
exists a sequence $\{t_l\}, t_l \rightarrow \infty$, and a constant $d>0$
such that for every $t_l$,
\begin{equation} \label{IQC}
\int_{0}^{t_l}\|\varphi(y(.) \mid_0^t)\|^2dt \leq \int_{0}^{t_l} \| y \|^2 dt + d, ~\forall y \in L_{2e} [0, \infty).
\end{equation}
The class of such operators will be denoted by $\Xi_0$.
\end{assumption}

%\begin{remark}\label{Rem1}
%This assumption captures some common classes of uncertain coupling. For
%example, $\varphi$ can be a linear causal operator from the Hardy space
%$H_\infty$. Such operators have extension to operator mapping $L_{2e} [0,
%\infty)$ into $L_{2e} [0, \infty)$ \cite{[Willems1971]}. For instance, it
%is easy to show that unmodelled dynamics described as
%\begin{displaymath}
%\left\{
%\begin{array}{ll}
%\dot{\zeta}_i=-a_i{\zeta}_i+ y(t),\\
%\varphi (y (.) \mid_0^t)=\zeta_i (t)
%\end{array}
%\right.
%\end{displaymath}
%satisfy (\ref{IQC}). Then the term $\varphi(x_j(.)-x_i(.))$ in (\ref{agents
%  dymamic}) can be interpreted as an action based on relative measurements
%and applied through a dynamic channel with memory.
%
%Uncertain input delay in receiving relative states is also allowed. Finally (\ref{IQC}) captures norm-bounded uncertain coupling
%$\varphi(y(.)|_0^t)=\Delta y(t)$ where ${\Delta}$ is a constant matrix such
%that ${\Delta}'\Delta \leq I$.
%\end{remark}

In addition to the system (\ref{agents dymamic}), suppose a leader is given. The dynamics of the leader, labeled $0$, is expressed as
\begin{equation} \label{leader dymamic}
 \dot{x}_0=Ax_0,
\end{equation}
where $x_0\in \Re^n$ is its state. We assume throughout the paper that the
leader node can be observed from a subset of nodes of the graph $\mathcal
{G}$. If the leader is observed by node $i$, we extend the graph
$\mathcal{G}$ by adding the edge $(0,i)$ with weighting gain $g_i=1$,
otherwise let $g_i=0$. We refer to node $i$ with $g_i\neq 0$ as a pinned or
controlled node. Denote the pinning matrix as $G=\mathrm{diag}\left[g_1,
  \cdots, g_N\right]\in \Re ^{N\times N}$. The system is assumed to have at
least one agent connected to the leader, hence $G \neq 0$.

Define synchronization error vectors as
$e_i=x_0-x_i$, $i=1,2,\ldots,N$.
Dynamics of these vectors satisfy the equation
\begin{equation} \label{error dymamic}
 \dot{e}_i=Ae_i-B_1u_i-B_2 \sum\limits_{j\in N_i}a_{ij} \varphi (e_i(.) \mid_0^t-e_j(.) \mid_0^t).
\end{equation}

In this paper we are concerned with finding a control protocol for each
node $i$ of the form
\begin{equation} \label{controller}
 u_i=-K\{\sum\limits_{j\in N_i}a_{ij}(x_j-x_i) + g_i(x_0-x_i)\},
\end{equation}
where $K \in \Re^{p\times n}$ is the feedback gain matrix to be found.
% The protocol (\ref{controller}) is to guarantee a certain tracking and
% consensus performance of the system;
% cf.~\cite{[Hongwei2011]}.
As a measure of the system performance under this protocol, we will
use the quadratic cost function (cf.~\cite{[Borrelli2008]}),
\begin{align} \label{cost function}
\mathcal {J}(u)=& \sum\limits_{i=1}^{N} \int_{0}^{\infty}\Big(\frac{1}{2}\sum\limits_{j\in N_i}(e_i - e_j)' Q (e_i - e_j) \nonumber\\
&+ g_i e_i' Q e_i +u_i'R u_i\Big)dt,
\end{align}
where $Q=Q'>0$ and $R=R'>0$ are given weighting matrices, and $u$ denotes
the vector $u=[u_1'~\ldots~u_N']'$.
Each addend in the cost function (\ref{cost function}) penalizes the $i$th
system input, the disagreement between the $i$th and the $j$th system states, as well as
the disagreement between the leader and the pinned agent which can
observe the leader. % A similar cost was considered, e.g., in
%  \cite{[Borrelli2008]} where the penalty was imposed on the disagreement
% between \emph{all} agents.

The problem in this paper is to find a control protocol (\ref{controller}) which solves the guaranteed performance leader following consensus control problem as follows:

\begin{problem}\label{prob1}
Under Assumption~\ref{A1}, find a control protocol of the form
(\ref{controller}) such that
\begin{equation} \label{cost function 1}
\sup \limits_{\Xi_0}\mathcal {J}(u) < \infty.
\end{equation}
\end{problem}

It will be shown later in the paper that (\ref{cost function 1}) implies
 \begin{equation} \label{synchronization}
\int_{0}^{\infty}\|e_i\|^2 dt
 < \infty \quad \forall i=1,\ldots,N.
\end{equation}
Hence, solving Problem~\ref{prob1} will guarantee that all the agents
synchronize to the leader in the $L_2$ sense.

\subsection{Associated Guaranteed Cost Decentralized Control Problem}
In this section, we introduce an auxiliary guaranteed cost decentralized control
problem for an interconnected large scale system. Our approach
follows~\cite{[Zhongkui2010],[Hongwei2011]}, however here it results in a
collection of interconnected subsystems.

Taking the linearity of the operator $\varphi$ into account, the closed
loop large-scale system consisting of the collection of the
error dynamics (\ref{error dymamic}) and the protocols
(\ref{controller}) can be written as
\begin{align}\label{large error dymamic}
\dot{e}=&(I_N\otimes A)e + ((\mathcal {L}+G)\otimes (B_1K))e  \nonumber\\
       & - ((\mathcal {L}+G)\otimes B_2)\Phi(t) + (G \otimes B_2)\Phi(t),
\end{align}
where
\begin{align*}
e&=[e_1'~ e_2'~\ldots~e_N']',\nonumber \\
\Phi(t)&=[\varphi'(e_1(.) \mid_0^t)~\varphi'(e_2(.) \mid_0^t)~\ldots~\varphi'(e_N(.) \mid_0^t)],
\end{align*}
and $\otimes$ denotes the Kronecker product.

It was shown in \cite{[Hong2006]} that if the communication graph
$\mathcal {G}$ is connected and has at least one agent connected to the
leader, then the symmetric matrix $\mathcal {L}+G$ is positive definite,
Hence all its eigenvalues are positive.

Let $T\in \Re^{N\times N}$ be an orthogonal matrix such that
\begin{equation} \label{transformation}
T^{-1}(\mathcal {L}+G)T=J=\mathrm{diag}\left[\lambda_1, \cdots, \lambda_N\right].
\end{equation}
Also, let $\varepsilon=(T^{-1}\otimes I_n)e$,
$\varepsilon=[\varepsilon_1'~\ldots~\varepsilon_N']'$ and $\Psi(t)=(T^{-1}
\otimes I_n)\Phi(t)$. Using this coordinate transformation, the system
(\ref{large error dymamic}) can be represented in terms of $\varepsilon$, as
\begin{align}
\label{error dymamic transformation}
\begin{split}
 \dot{\varepsilon} =& \big({I_N\otimes A + J\otimes (B_1K)}\big)\varepsilon -\big(J\otimes B_2)\Psi(t) \\
 &- \big((T^{-1}GT) \otimes B_2\big)\Psi(t),
\end{split}
\end{align}
where
\begin{align}
\Psi(t)=\left(\begin{array}{ccc}
                           \varphi(\sum \limits_{j=1}^{N}(T^{-1})_{1j} e_j(.) \mid_0^t)\\
                           \vdots\\
                           \varphi(\sum \limits_{j=1}^{N}(T^{-1})_{Nj} e_j(.) \mid_0^t)
                        \end{array}\right)=\left(\begin{array}{ccc}
                           \varphi(\varepsilon_1(.) \mid_0^t)\\
                           \vdots\\
                           \varphi(\varepsilon_N(.) \mid_0^t)
                        \end{array}\right).\nonumber
\end{align}
Here we used the assumption that $\varphi(.)$ is a
linear operator. It follows from (\ref{error dymamic transformation}) that
 \begin{align} \label{another error dymamic}
 \dot{\varepsilon}_i &= A\varepsilon_i + \lambda_i B_1K\varepsilon_i +\sum\limits_{j\neq i}\sum\limits_{p} (T^{-1})_{ip} g_{p}T_{pj} B_2\varphi(\varepsilon_j(.) \mid_0^t)\nonumber\\
                     &+(\sum\limits_{p}  (T^{-1})_{ip} g_{p}T_{pi} - \lambda_i)B_2\varphi(\varepsilon_i(.) \mid_0^t) ,
 \end{align}
Then (\ref{error dymamic transformation}) can be regarded as a closed loop
system consisting of $N$ interconnected linear uncertain subsystems of the
following form
 \begin{equation} \label{simu error dymamic1}
 \dot{\varepsilon}_i=A\varepsilon_i + B_{1i}\hat{u}_i + E_{i}\xi_i + L_i
 \eta_{i},
 \end{equation}
each governed by a state feedback controller $ \hat{u}_i= K\varepsilon_i$.
Here we have used the following notation
 \begin{align}
\xi_i&= \varphi(\varepsilon_i (.) \mid_0^t), \label{xi.i} \\
\eta_{i}&= [\xi_1'~ \ldots~
\xi_{i-1}'~\xi_{i+1}'~\ldots~\xi_N']', \label{eta.i} \\
B_{1i}&={\lambda}_i B_{1}, \nonumber \\
E_i&=(f_{i,i} - \lambda_i)B_2,\nonumber  \\
L_i&=B_2[f_{i,1}I,\ldots, f_{i,(i-1)}I,f_{i,(i+1)}I,\ldots,f_{i,N}I],\nonumber  \\
f_{i,j}&=\sum\limits_{p} {T^{-1}}_{ip} g_{p}T_{pj}. \nonumber
\end{align}

According to Assumption~\ref{A1}, the linear coupling $\varphi$ satisfies the
integral quadratic constraints (\ref{IQC}). This implies that the following
two inequalities hold for all $i=1, 2, \cdots, N$:
\begin{align}
\label{IQC 1}
& \int_{0}^{t_l} \|\xi_i \|^2dt \leq \int_{0}^{t_l} \| \varepsilon_i \|^2 dt
+ d, \\
\label{IQC 2}
& \int_{0}^{t_l} \|\eta_{i}\|^2 dt \leq \int_{0}^{t_l} \sum\limits_{j\neq
  i}\|\varepsilon_j\|^2dt + (N-1)d.
\end{align}
It follows from (\ref{IQC 1}) and (\ref{IQC 2}) that the collection of
uncertainty inputs $\xi_i$, $\eta_i$, $i=1, 2, \cdots, N$, represents an
admissible local uncertainty and admissible interconnection inputs for the
large-scale system (\ref{simu error dymamic1}); see
\cite{[Ugrinovskii2005],[Li2007],[Ian2000],[Ugrinovskii2000]}.  The sets of
admissible uncertainty inputs and admissible interconnection inputs for the
system (\ref{simu error dymamic1}) will be denoted by $\Xi$, $\Pi$,
respectively. Then the above discussion can be summarized as follows: If
$\varphi$ satisfies the conditions in Assumption~\ref{A1}, then the corresponding
signals (\ref{xi.i}), (\ref{eta.i}) belong to $\Xi$, $\Pi$, respectively.

Next, consider the performance cost (\ref{cost function}). It is possible to show that
\begin{align}
\label{***}
\mathcal {J}(u) &= \int_{0}^{\infty} \big(e' ((\mathcal {L}+G)\otimes Q) e +
{u}'(I\otimes R){u} \big) dt. \\
% \end{align}
%
% According to the coordinate transformation (\ref{transformation}), $T$ is
% an orthogonal matrix, $\varepsilon=(T^{-1}\otimes I)e$,  then $e =(T
% \otimes I)\varepsilon $
% and
% \begin{align*}
% u =-[(\mathcal {L}+G)T\otimes K]\varepsilon.
% \end{align*}
% This leads to the following expression for the performance cost
% \begin{align*}
% \mathcal {J}(u)&=\int_{0}^{\infty} \Big(\varepsilon' \big( T'(\mathcal {L}+G)T\otimes Q\big)
% \varepsilon \\
% &+ \varepsilon' \big(T'(\mathcal {L}+G)(\mathcal {L}+G)T\otimes K' R K \big)\varepsilon \Big) dt.  \nonumber
% \end{align*}
% Since $T'T=T T'=I$, then we can continue as follows
% \begin{align}
% \mathcal {J}(u)
&= \sum_{i=1}^{N} \int_{0}^{\infty} \big(\lambda_i\varepsilon_i'  Q
\varepsilon_i + \lambda_i^{2} {\varepsilon}_i'K' R K {\varepsilon_i}
\big) dt.
\label{cost function L2}
\end{align}
Thus we conclude that
\begin{equation}\label{equivalent cost function 1}
 \mathcal{J}(u)=  \hat{\mathcal{J}}(\hat{u}),
\end{equation}
where $\hat u=(\hat u_1',\ldots,\hat u_N')'$, $\hat u_i=K\varepsilon_i$, and
\begin{equation}\label{Jhat}
\hat{\mathcal{J}}(\hat{u})=\sum_{i=1}^{N} \int_{0}^{\infty}
\big(\lambda_i \varepsilon_i' Q \varepsilon_i + \lambda_i^{2}
\hat u_i R \hat u_i \big) dt,
\end{equation}

Now consider the auxiliary guaranteed cost decentralized control problem
associated with the uncertain large scale system comprised of the
subsystems (\ref{simu error dymamic1}), with uncertainty inputs
(\ref{xi.i}) and interconnections (\ref{eta.i}), subject to the IQCs
(\ref{IQC 1}), (\ref{IQC 2}). In this problem we wish to find a
decentralized state feedback controller $\hat u=(\hat u_1',\ldots,\hat
u_N')'$, $\hat{u}_i=K\varepsilon_i$ such that
\begin{equation}\label{sup.Jhat}
\sup_{\Xi,\Pi} \hat{\mathcal{J}}(\hat{u}) < \infty.
\end{equation}
The discussion in this section can be summarized as follows.

\begin{lemma}\label{lem2}
Under Assumption~\ref{A1}, if the state feedback controller
$\hat u=(\hat u_1',\ldots,\hat u_N')'$, $\hat{u}_i=K\varepsilon_i$, solves
the auxiliary guaranteed cost decentralized control problem for the
collection of systems (\ref{simu error dymamic1}) and the cost function
(\ref{Jhat}),  then the control
protocol (\ref{controller}) with the gain matrix $K$ solves
Problem~\ref{prob1}.
\end{lemma}

\begin{proof}
Since the state feedback controller $\hat{u}_i=K\varepsilon_i$ solves the
auxiliary guaranteed cost decentralized control problem for the collection
of the systems (\ref{simu error dymamic1}), then we have
\begin{equation}\label{}
\sup_{\Xi,\Pi} \hat{\mathcal{J}}(\hat{u}) < c.
\end{equation}
Also we noted that every signal $\varphi$ which satisfies Assumption~\ref{A1}
gives rise to an admissible uncertainty for the large-scale system
consisting of subsystems (\ref{simu error
  dymamic1}). This implies that for any $\varphi \in \Xi_0$ with
$u=(\mathcal{L}+G)\otimes Ke$, we have
\begin{equation}\label{equivalent cost function 2}
\mathcal{J}(u) = % \int_{0}^{\infty}\sum\limits_{i=1}^{N}\|z_i\|^2 dt =
\hat{\mathcal{J}}(\hat{u}) \leq \sup_{\Xi, \Pi} \hat{\mathcal{J}}(\hat{u})< c.
\end{equation}
It implies that  the control protocol
(\ref{controller}) with the same gain $K$ solves Problem~\ref{prob1}.
\end{proof}

Note that since $\mathcal{L}+G$ is positive definite, then
(\ref{synchronization}) follows from
(\ref{equivalent cost function 1}) and (\ref{***}).

\section{The Main Results}

The main results of this paper are
sufficient conditions under which the system (\ref{agents dymamic}),
(\ref{controller})  tracks
the leader with guaranteed consensus tracking performance.

\begin{theorem}\label{Theorem 1}
Let matrices $Y=Y'>0, Y \in \Re^{n\times n} $ and $F\in \Re^{p\times n}$, and constants ${\pi}_i>0$, ${\theta}_i>0$, $i=1, 2,\cdots, N$, exist such that the following LMIs are satisfied simultaneously
\begin{equation}
\label{LMI TH1}
\left[
\begin{array}{ccccc}
Z_i   & F'  & YQ_i ^{1/2}  & Y& \mathbf{1}'\otimes Y \\
F  & -\frac{1}{\lambda_i^2}R^{-1}  & 0  & 0 &        0  \\
 Q_i^{1/2}Y & 0  & -I & 0 & 0 \\
 Y & 0 & 0 & -\frac{1}{{\pi}_i}I & 0  \\
 \mathbf{1}\otimes Y & 0 & 0 & 0  & -\Theta_i^{-1}
\end{array}\right]<0,
\end{equation}
 where
 \begin{align*}
 \mathbf{1}&=[1~\ldots~1]'\in\Re^{N-1},\\
 p_i &= f_{i,i}-\lambda_i, \\
 q_i &= \big(\sum \limits_{j\neq i} f_{i,j}^2\big)^{1/2},\\
 Q_i &=\lambda_iQ, \\
 \end{align*}
 and
\begin{align*}
&Z_i=A Y + Y A' +  \lambda_iF'B_1'+\lambda_iB_1 F\\
 &+({\pi}_i^{-1}(p_i)^2+ {\theta}_i^{-1} q_{i}^2)B_2B_2',\\
&\Theta_i=\mathrm{diag}[\theta_1,\ldots,\theta_{i-1},
\theta_{i+1},\ldots,\theta_N] \otimes I_{n}.
\end{align*}
Then the control protocol (\ref{controller}) with $K= FY^{-1}$ solves
Problem~\ref{prob1}. Furthermore, this protocol guarantees the following
performance bound
 \begin{equation} \label{cost function TH1}
 \sup_{\Xi_0}\mathcal {J}(u) \le \sum_{i=1}^{N} ({\pi}_i+
 {\theta}_i (N-1)) d + \sum_{i=1}^{N} e_i'(0)Y^{-1}e_i(0).
\end{equation}
\end{theorem}

\begin{proof}
 Using the Schur complement, the LMIs (\ref{LMI TH1}) can be transformed into the following Riccati inequality
\begin{align}\label{riccati inequality}
&A Y + Y A'  + \Big[\frac{p_i^2}{{\pi}_i} + \frac {q_{i}^2}{{\theta}_i}\Big]B_2B_2' + Y[ Q_i+ ({\pi}_i + {\bar{\theta}}_i) I ] Y \nonumber\\
&\quad + \lambda_i^2 F'R F + \lambda_iF'B_1' + \lambda_iB_1 F < 0,
\end{align}
where ${\bar{\theta}}_i=\sum \limits_{j\neq i} {\theta}_i$.

Consider the following Lyapunov function candidate for the interconnected
systems (\ref{simu error dymamic1}), (\ref{IQC 1}), (\ref{IQC 2}):
\begin{equation}
V(\varepsilon)=\sum_{i=1}^{N}\varepsilon_i'Y^{-1}\varepsilon_i.
\end{equation}
For the controller $\hat{u}_i=K \varepsilon_i$, we have
\begin{align}
\label{lyapunov equation1}
&\frac{d V(\varepsilon)}{dt}= 2 \sum_{i=1}^{N}\varepsilon_i'Y^{-1}p_{i}B_2\xi_i+ 2\sum_{i=1}^{N}\varepsilon_i'Y^{-1}L_i \eta_{i} \\
&+ \sum_{i=1}^{N}\varepsilon_i'\Big[Y^{-1}(A+\lambda_i B_1K)+(A+\lambda_i B_1K)'Y^{-1}\Big]\varepsilon_i. \nonumber
\end{align}
Substitute $F=KY$ into the Riccati inequality (\ref{riccati inequality}). Then after pre- and post-multiplying (\ref{riccati inequality}) by $Y^{-1}$ and substituting it into (\ref{lyapunov equation1}), we have
\begin{align*}
&\frac{d V(\varepsilon)}{dt}<  - \sum_{i=1}^{N} \Big(\varepsilon_i' {\pi}_i^{-1} p_i^{2}Y^{-1}B_2B_2'Y^{-1} \varepsilon_i \\
& +2 \varepsilon_i'Y^{-1}p_{i}B_2\xi_i - {\pi}_i \parallel\xi_i\parallel^{2} + {\pi}_i\parallel\xi_i\parallel^{2} -{\pi}_i \parallel\varepsilon_i\parallel^{2} \\
&   - \varepsilon_i'{\theta}_i^{-1} q_i^{2} Y^{-1}B_2B_2'Y^{-1} \varepsilon_i +2\varepsilon_i'Y^{-1}L_i \eta_{i} - {\theta}_i \parallel\eta_i\parallel^{2} \\
&  + {\theta}_i \parallel\eta_i\parallel^{2} -{\bar{\theta}}_i \parallel\varepsilon_i\parallel^{2} \Big) -\sum_{i=1}^{N}\varepsilon_i' \Big[\lambda_i^2 K'RK + Q_i \Big]\varepsilon_i. \nonumber
\end{align*}
Using the following identity,
\begin{align}\label{identity}
\sum\limits_{i=1}^{N} \sum\limits_{j\neq i} {\theta}_{i} \|\varepsilon_j\|_2^2 =\sum\limits_{i=1}^{N}{\bar{\theta}}_{i}\|\varepsilon_i\|_2^2,
\end{align}
one has
\begin{align}
\int_{0}^{t_l}\frac{d V(\varepsilon)}{dt}dt &<-\sum_{i=1}^{N}\int_{0}^{t_l}\varepsilon_i' (\lambda_i^2 K'RK + Q_i )\varepsilon_i dt \nonumber\\
&+ \sum_{i=1}^{N} {\pi}_i\int_{0}^{t_l}(\parallel\xi_i\parallel^{2}-\parallel\varepsilon_i\parallel^{2})dt \nonumber\\
& +\sum_{i=1}^{N}{\theta}_i\int_{0}^{t_l}(\parallel\eta_i\parallel^{2}-\sum
\limits_{j\neq i} \parallel\varepsilon_j\parallel^{2})dt.  \nonumber
\end{align}
Finally, using the IQCs (\ref{IQC 1}) and (\ref{IQC 2}) and $V(\varepsilon(t)) > 0 $, we obtain
\begin{align}
\sum_{i=1}^{N} \int_{0}^{t_l}\varepsilon_i' (\lambda_i^2 K'R K + Q_i
)\varepsilon_i dt <& \sum_{i=1}^{N}({\pi}_i  +  {\theta}_i (N-1))d\nonumber\\
&+ V(\varepsilon(0)).
\end{align}
The expression on the right hand side of the above inequality is independent of $t_l$. Letting $t_l \rightarrow \infty$ leads to
\begin{align}
\label{cost1}
\mathcal{J}(\hat u)\leq \sum_{i=1}^{N} ({\pi}_i + \sum_{i=1}^{N} {\theta}_i (N-1)) d + V(\varepsilon(0)).
\end{align}
Also (\ref{cost1}) holds for arbitrary collection of inputs $\xi_i$,
$\eta_i$ that satisfy (\ref{IQC 1}), (\ref{IQC 2}), respectively. Then
we conclude that
\begin{align}
\label{cost1.2}
\sup_{\Xi,\Pi} \hat{\mathcal{J}}(\hat u)\le
\sum_{i=1}^{N} (({\pi}_i + \sum_{i=1}^{N}{\theta}_i (N-1)) d  +
  e_i'(0)Y^{-1}e_i(0)).
\end{align}
The result of the theorem now follows from Lemma~\ref{lem2} and
(\ref{equivalent cost function 2}).
\end{proof}

According to Theorem~\ref{Theorem 1}, one has to solve $N$ coupled LMIs to obtain the
control gain $K$. To simplify the calculation, we establish the
following theorem which requires only one LMI to be feasible, as follows

\begin{equation}
\label{LMIT3}
\left[
\begin{array}{cccccccccc}
\bar Z            & Y (\bar\lambda Q) ^{1/2}  & Y     &   Y\\
(\bar\lambda Q) ^{1/2}Y        &    -I       &             0             &          0            \\
 Y  &    0        &  -\frac{1}{\pi}I          &          0            \\
 Y  &    0        &             0             & -\frac{1}{(N-1){\theta}}I
\end{array}\right]<0,
\end{equation}
where
\begin{align*}
p^2 =\max \limits_i(p_i)^2, \\
q^2=\max \limits_i(q_{i}^2),\\
\underline{\lambda} = \min \limits_i(\lambda_i), \\
\bar{\lambda} = \max \limits_i(\lambda_i),
\end{align*}
and
\begin{align*}
\bar Z= A Y + Y A' - \frac{\underline{\lambda}^2}{\bar \lambda^2}
B_1R^{-1}B_{1}' + \Big[\frac{p^2}{{\pi}} + \frac
{q^2}{{\theta}}\Big]B_2B_2' .
\end{align*}
\begin{lemma}\label{Lemma 3}
Given $R=R'>0$ and $Q> 0$, suppose the LMI
(\ref{LMIT3}) in variables $Y=Y'>0$, $\pi^{-1}>0$ and $\theta^{-1}>0$ is feasible.
Then the matrices and constants
\begin{align}
Y,\quad F=-\frac{\underline{\lambda}}{{\bar\lambda}^2} R^{-1}B_1', \quad
\pi_i=\pi, \quad \theta_i=\theta
\end{align}
are a feasible set of matrices and constants for the collection of LMIs
(\ref{LMI TH1}).
\end{lemma}
\begin{proof}
Using the Schur complement, the above LMI (\ref{LMIT3}) is equivalent to the following Riccati inequality
\begin{align}
\label{one inequality}
&A Y + Y A' - \frac{\underline{\lambda}^2}{\bar\lambda^2}
B_1R^{-1}B_{1}' + \Big[\frac{p^2}{{\pi}} + \frac {q^2}{{\theta}}\Big]B_2B_2' \nonumber\\
 &+ Y [\bar\lambda Q + ({\pi} + \bar{\theta}) I ] Y < 0,
\end{align}
where $\bar{\theta}=(N-1){\theta}$.

Since $\underline{\lambda}\le \lambda_i \le \bar\lambda $, (\ref{one
  inequality}) implies that
\begin{align}
\label{one inequality22}
&A Y + Y A' - \lambda_i \frac{\underline{\lambda}}{\bar\lambda^2}
B_1R^{-1}B_{1}' - \lambda_i \frac{\underline{\lambda}}{\bar\lambda^2} B_1R^{-1}B_{1}' \nonumber \\
& + \frac{\underline{\lambda}^2}{\bar\lambda^2} B_1R^{-1}B_{1}' + \Big[\frac{p^2}{{\pi}} + \frac {q^2}{{\theta}}\Big]B_2B_2' \nonumber \\
& + Y [\lambda_i Q + ({\pi} + \bar{\theta}) I ] Y< 0.
\end{align}
Substitute %$K=- \underline{\lambda} R^{-1}B_{1}'Y^{-1}$
$F=-\frac{\underline{\lambda}}{{\bar\lambda}^2} R^{-1}B_1'$,
$\pi_i=\pi$, $\theta_i=\theta$, $\bar\theta_i=(N-1)\theta=\sum \limits_{j\neq
  i}\theta_i$,  and let $ p_i^{2}=(f_{i,i}-\lambda_i)^2
\leq p^2$ and $q_i^2=(\sum \limits_{j\neq i} f_{i,j}^2)\leq q^2$, then we
obtain
% \begin{equation}
% \label{one inequality2}
% %\begin{split}
%  (A + \lambda_i B_1K) Y + Y (A+\lambda_i B_1K)' + Y K'R K Y
%  + \Big[\frac{p_i^2}{\pi} + \frac {q_i^2}{{\theta}}\Big]B_2B_2' + Y [Q + (\pi + \bar{\theta}) I ] Y < 0.
%  %\end{split}
% \end{equation}
\begin{equation}
\label{one inequality2}
\begin{split}
&\lefteqn{A Y + Y A' + \lambda_i F'B_{1}' + \lambda_i B_1F
+ \bar\lambda^2 F'R^{-1}F}  \\
&+ \Big[\frac{p_i^2}{{\pi_i}} + \frac {q_i^2}{{\theta_i}}\Big]B_2B_2'
 + Y [\lambda_i Q + ({\pi_i} + \bar{\theta_i}) I ] Y< 0.
\end{split}
\end{equation}
Noting that $\lambda_i^2 F'R^{-1}F\le \bar\lambda^2 F'R^{-1}F$, we obtain
the Riccati inequality (\ref{riccati inequality}) which is equivalent to
(\ref{LMI TH1}).
\end{proof}

\begin{theorem}\label{Theorem 2}
Given $R=R'>0$ and $Q> 0$, suppose the LMI
(\ref{LMIT3}) in variables $Y=Y'>0$, $\pi^{-1}>0$ and $\theta^{-1}>0$ is
feasible. Then the control protocol (\ref{controller}) with $K=
-\frac{\underline \lambda}{\bar \lambda^2} R^{-1}B_1'Y^{-1}$ solves
Problem~\ref{prob1}. Furthermore, this protocol guarantees the following
performance bound
 \begin{equation} \label{cost function TH2}
 \sup_{\Xi_0}\mathcal {J}(u) \le N({\pi}+
 {\theta} (N-1)) d + \sum_{i=1}^{N} e_i'(0)Y^{-1}e_i(0).
\end{equation}
\end{theorem}
\begin{proof}
The proof readily follows from Lemma~\ref{Lemma 3} and Theorem~\ref{Theorem 1}.
\end{proof}

\section{Example}
To illustrate the proposed method, consider a system consisting of three
identical pendulums coupled by two springs. Each pendulum is subject to an
input as shown in Fig.~\ref{pendulums}. The dynamic of the coupled system
is governed by the following equations
\begin{align}
\label{dynamic of pedulums}
ml^2\ddot{\alpha}_1=&-ka^2(\alpha_1-\alpha_2)-mgl\alpha_1-u_1,  \nonumber\\
ml^2\ddot{\alpha}_2=&-ka^2(\alpha_2-\alpha_3)-ka^2(\alpha_2-\alpha_1)\\
 &-mgl\alpha_2-u_2, \nonumber\\
ml^2\ddot{\alpha}_3=&-ka^2(\alpha_3-\alpha_2)-mgl\alpha_3-u_3,\nonumber
\end{align}
where $l$ is the length of the pendulum, $a$ is the position of the spring, $g$ is the gravitational acceleration constant, $m$ is the mass of each pendulum, and  $k$ is the spring constant.

In addition to the three pendulums, consider the leader pendulum which is
identical to those given. Its dynamics are described by the equation
\begin{equation}
\label{leader pedulums}
ml^2\ddot{\alpha}_0=-mgl\alpha_0.
\end{equation}
Choosing the state vectors as $x_0=(\alpha_0, \dot{\alpha}_0)$, $x_1=(\alpha_1, \dot{\alpha}_1)$, $x_2=( \alpha_2, \dot{\alpha}_2)$ and $x_3=(\alpha_3, \dot{\alpha}_3)$, the equation (\ref{dynamic of pedulums}) and (\ref{leader pedulums}) can be written in the form of (\ref{agents dymamic}), (\ref{leader dymamic})
,
where
\begin{align*}
A=\left[\begin{array}{cc}
                          0  &  1  \\
                        -\frac{g}{l} &  0
                        \end{array}\right],
                        \end{align*}
                        \begin{align*}
                         B_1=\left[\begin{array}{c}
                          0    \\
                        -\frac{1}{ml^2}
                        \end{array}\right],
                        \end{align*}
                        \begin{align*}
                         B_2=\left[\begin{array}{cc}
                                0        &  0\\
                         \frac{a^2}{ml^2}&  0
                        \end{array}\right],
                        \end{align*}
                         and
                         \begin{align*}
                         \varphi(x_j-x_i)=k(x_j-x_i).
                         \end{align*}

The parameters of the coupled pendulum system are chosen as $m=0.25kg$, $l=1m$, $a=0.5m$, $g=10m/s^2$, $ 0 \leq k \leq 1 N/m$.

\begin{figure}
\centering
\begin{tikzpicture}[scale=1.1,
    media/.style={font={\footnotesize\sffamily}},
    wave/.style={
        decorate,decoration={snake,post length=1.4mm,amplitude=2mm,
        segment length=2mm},thick},
    interface/.style={
        postaction={draw,decorate,decoration={border,angle=-45,
                    amplitude=0.3cm,segment length=2mm}}},
    gluon/.style={decorate, draw=black,
        decoration={coil,amplitude=4pt, segment length=5pt}},scale=0.7
    ]

 %   \tikzstyle{damper}=[thick,decoration={markings,
%   mark connection node=dmp,
%   mark=at position 0.5 with
%   {
%     \node (dmp) [thick,inner sep=0pt,transform shape,rotate=-90,minimum
% width=15pt,minimum height=3pt,draw=none] {};
%     \draw [thick] ($(dmp.north east)+(2pt,0)$) -- (dmp.south east) -- (dmp.south
% west) -- ($(dmp.north west)+(2pt,0)$);
%     \draw [thick] ($(dmp.north)+(0,-5pt)$) -- ($(dmp.north)+(0,5pt)$);
%   }
% }, decorate]
%
% \draw[damper] ($(-1.3,-3.2) + (0,0.5)$) -- ($(1.7,-3.2) + (0,0.5)$);
%  \draw[damper] ($(1.7,-3.2) + (0,0.5)$) -- ($(4.7,-3.2) + (0,0.5)$);
%
%  \filldraw[fill=white,line width=1pt](-1.32,-2.7) circle(.02cm);
%     \filldraw[fill=white,line width=1pt](1.67,-2.7)  circle(.02cm);
%     \filldraw[fill=white,line width=1pt](4.67,-2.7)circle(.02cm);
    % Round rectangle
    %%%%\fill[gray!10,rounded corners] (-4,-3) rectangle (4,0);
    % Interface
    %%%%%\draw[blue,line width=.5pt,interface](-6,0)--(6,0);
    \draw[blue,line width=1pt](-5.2,0)--(4.2,0);
    % Vertical dashed line
    \draw[dashed,black](-5,-3)--(-5,0);
    \draw[dashed,black](-2,-3)--(-2,0);
    \draw[dashed,black](1,-3)--(1,0);
    \draw[dashed,black](4,-3)--(4,0);

  \draw[-,gluon]
        (28:-2cm)--(141.5:-1.55cm);
  \draw[-,gluon]
        (141.5:-1.55cm)--(167.5:-4.35cm);
    \draw[black](0:-5cm)--(46:-5.7cm)node[midway]{\scriptsize $~~l$};
    \path (-5,0)++(-84:2cm)node{\scriptsize  $\alpha_0$};

    \draw[->] (-5,-1.5) arc (-90:-63:.75cm);
    \draw[black](0:-2cm)--(76:-4.1cm)node[very near start]{\scriptsize $~~a$};
    \path (-2,0)++(-84:2cm)node{\scriptsize  $\alpha_1$};

    \draw[->] (-2,-1.5) arc (-90:-63:.75cm);
    \draw[black](0:1cm)--(116:-4.6cm);
    \path (1,0)++(-84:2cm)node{\scriptsize  $\alpha_2$};

    \draw[->] (1,-1.5) arc (-90:-63:.75cm);
    \draw[black](0:4cm)--(141:-6.4cm);
    \path (4,0)++(-84:2cm)node{\scriptsize $\alpha_3$};

    \draw[->] (4,-1.5) arc (-90:-63:.75cm);

    \filldraw[fill=white,line width=1pt](-1.77,-0.96)circle(.02cm);
    \filldraw[fill=white,line width=1pt](1.23,-0.96)circle(.02cm);
     \filldraw[fill=white,line width=1pt](4.23,-0.96)circle(.02cm);
    \filldraw[fill=white,line width=1pt](-4,-4)circle(.2cm)node[right]{\scriptsize $~Leader$};
       \filldraw[fill=white,line width=1pt](-1,-4)circle(.2cm)node[right]{\scriptsize $~\leftarrow u_1$};
          \filldraw[fill=white,line width=1pt](2,-4)circle(.2cm)node[left]{\scriptsize $u_2 \rightarrow~$};
             \filldraw[fill=white,line width=1pt](5,-4)circle(.2cm)node[left]{\scriptsize $ u_3\rightarrow~$};
\end{tikzpicture}
\caption{Interconnected pendulums.}
\label{pendulums}
\end{figure}
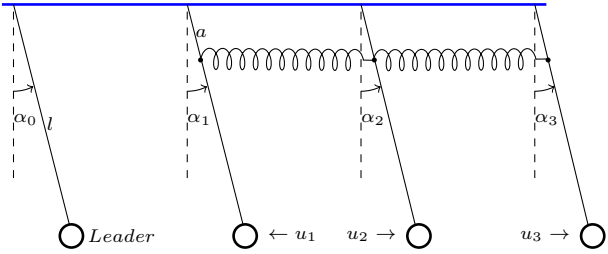

The communication topology of the coupled system is an undirected graph, shown in Fig.~\ref{Leader}. Suppose agent $1$ is equipped with sensors which allow it to observe the leader. Then the Laplacian matrix of the graph consisting of nodes 1, 2 and 3 and the pinning matrix $G$ can be described by
\begin{align*}
\mathcal {L}=\left[\begin{array}{ccc}
                         1    &    -1    &   0\\
                         -1     &    2    &   -1\\
                         0     &    -1    &   1
                        \end{array}\right],
\end{align*}
\begin{align*}
                        G=\left[\begin{array}{ccc}
                         1     &    0    &   0\\
                         0     &    0    &   0\\
                         0     &    0    &   0
                        \end{array}\right].
\end{align*}

It is a straightforward to verify that $\mathcal {L}+G$ is a positive definite symmetric matrix.

Two simulations are implemented for Theorem 1 and Theorem 2,
respectively. We use the same initial conditions and the matrixes $Q=I$ and
$R=0.1$. A computational algorithm based on Theorem 1 yields the gain matrix $K=[3.9870,
4.5178]$. For this gain, we
directly computed the cost function (\ref{cost
  function}) to yield $\mathcal{J}(u)=2.1018$.

The second simulation is based on Theorem~\ref{Theorem 2}. It involves
solving the decoupled
LMIs. Using the same matrices
$Q$ and $R$, and the same initial conditions, the control gain matrix was
computed to be $K=[1.3758,
10.4192]$, and the cost function was computed
numerically to be equal $\mathcal{J}(u)=4.9072$.

In this example, the conservatism is manifested by a greater
synchronization time. Compared with
the method based on Theorem 2, the method based
on Theorem 1 enables the  follower to synchronize to the leader in a
much shorter time,  with a better guaranteed performance.
 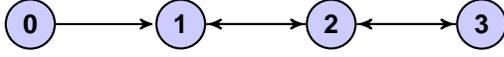
\begin{figure}
 \centering
 \begin{tikzpicture}[ ->,>=stealth',shorten >=1pt,auto,node distance=2cm,
   %thick,main node/.style={circle,fill=blue!20,draw,font=\sffamily\Large\bfseries}]
   thick,main node/.style={circle,fill=blue!20,draw,font=\sffamily\bfseries},]

  % \node[main node] (1) {0};
 %  \node[main node] (2) [below left of=1] {1};
 %  \node[main node] (3) [below right of=2] {3};
 %  %\node[main node] (4) [below left of=1] {3};
 %  \node[main node] (5) [below right of=1] {2};

    \node[main node] (1) {0};
   \node[main node] (2) [right of=1] {1};
   \node[main node] (3) [right of=2] {2};
   %\node[main node] (4) [below left of=1] {3};
   \node[main node] (5) [right of=3] {3};

   \path[every node/.style={font=\sffamily\small}]
     (1) %edge node [left] {0.6} (4)
        edge [right] node[left] {} (2)
        %edge [bend right] node[left] {0.3} (2)
         %edge [loop above] node {0.1} (1)
     (2) %edge node [right] {} (1)
       edge node {} (3)
        % edge [loop left] node {0.4} (2)
        % edge [right] node[left] {0.1} (3)%edge [bend right] node[left] {0.1} (3)
     (5) edge node [right] {} (3)
         %edge [right] node[right] {0.2} (4)%edge [bend right] node[right] {0.2} (4)
     %(4) edge node [left] {} (2)
         % edge [loop right] node {0.6} (4)
         % edge [bend right] node[right] {0.2} (1);
     (3) %edge node [right] {} (1)
        edge node {} (2)
        % edge [loop right] node {0.6} (4)
        edge [right] node[right] {} (5);  %edge [bend right] node[right] {0.2} (3);
 \end{tikzpicture}
 \centering
 \caption{Communication graph.}
 \label{Leader}
 \end{figure}

%
%\begin{figure}[htbp]
%\centering
%\subfigure{
%\includegraphics[width=0.232\textwidth]{Theorem1S1.eps}}
%\subfigure{
%\includegraphics[width=0.232\textwidth]{Theorem1S2.eps}}
%\caption{Relative angles (the left figure) and relative velocities of the
%  agents with
%  respect to the leader, obtained using the algorithm based on
%  Theorem~\ref{Theorem 1}.}
%\label{Theorem1state1}
%\end{figure}
%
%
%\begin{figure}[htbp]
%\centering
%\subfigure{
%\includegraphics[width=0.232\textwidth]{Theorem2S1.eps}}
%\subfigure{
%\includegraphics[width=0.232\textwidth]{Theorem2S2.eps}}
%\caption{Relative angles (the left figure) and relative velocities of the
%  agents with respect to the leader, obtained using the algorithm based on
%  Theorem~\ref{Theorem 2}.}
%  \label{Theorem2state1}
%\end{figure}

%\begin{figure}[htbp]
%\centering
%\includegraphics[width=.98\columnwidth]{Theorem1S1.eps}\\
%\includegraphics[width=.98\columnwidth]{Theorem1S2.eps}
%\caption{Relative angles (the top figure) and relative velocities of the
%  agents with
%  respect to the leader, obtained using the algorithm based on
%  Theorem~\ref{Theorem 1}.}
%\label{Theorem1state1}
%\end{figure}
%
%\begin{figure}[htbp]
%\centering
%\includegraphics[width=.98\columnwidth]{Theorem2S1.eps}\\
%\includegraphics[width=.98\columnwidth]{Theorem2S2.eps}
%\caption{Relative angles (the top figure) and relative velocities of the
%  agents with respect to the leader, obtained using the algorithm based on
%  Theorem~\ref{Theorem 2}.}
%\label{Theorem2state1}
%\end{figure}

\begin{figure}[htbp]
\centering
\includegraphics[width=.98\columnwidth]{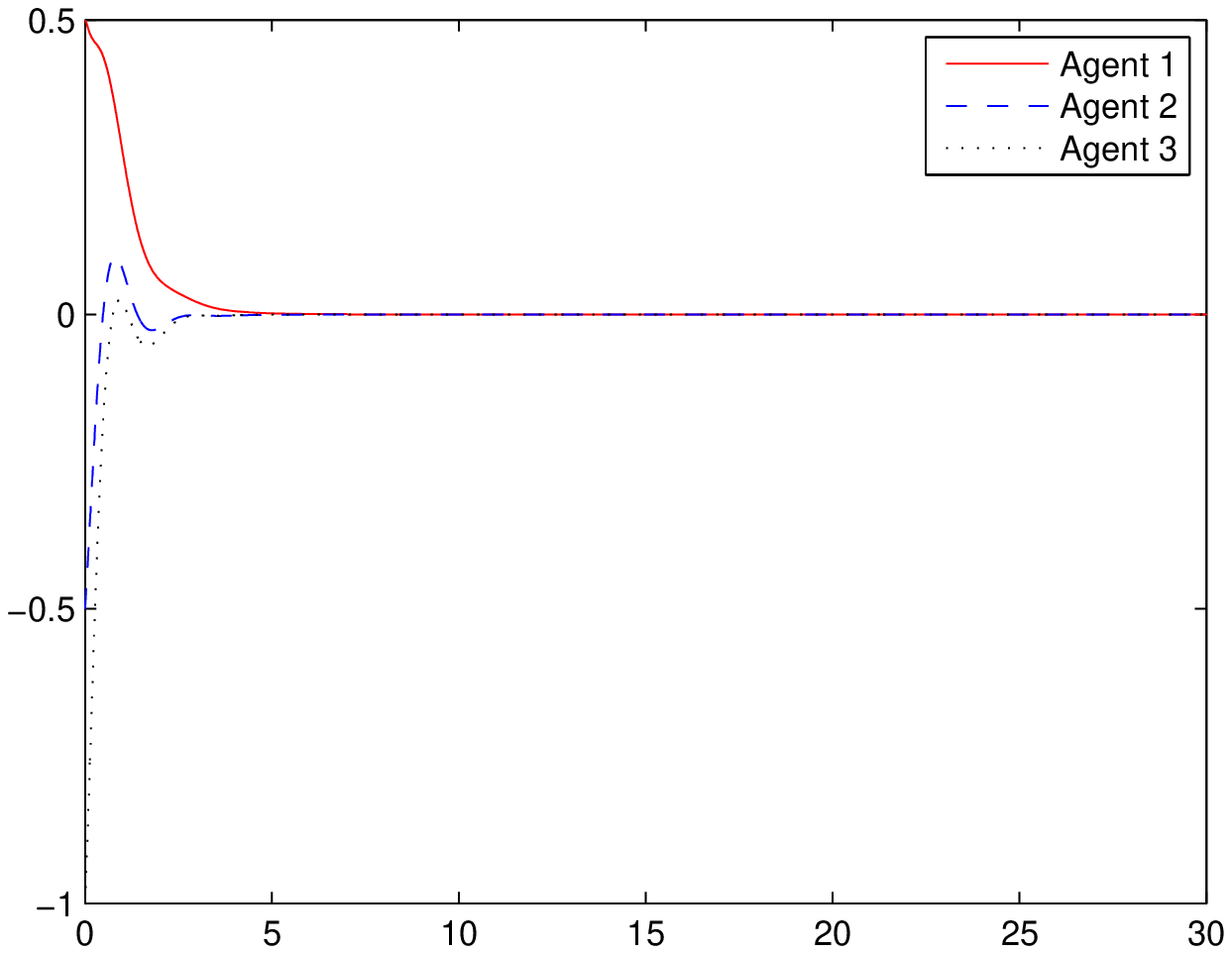}\\
\includegraphics[width=.98\columnwidth]{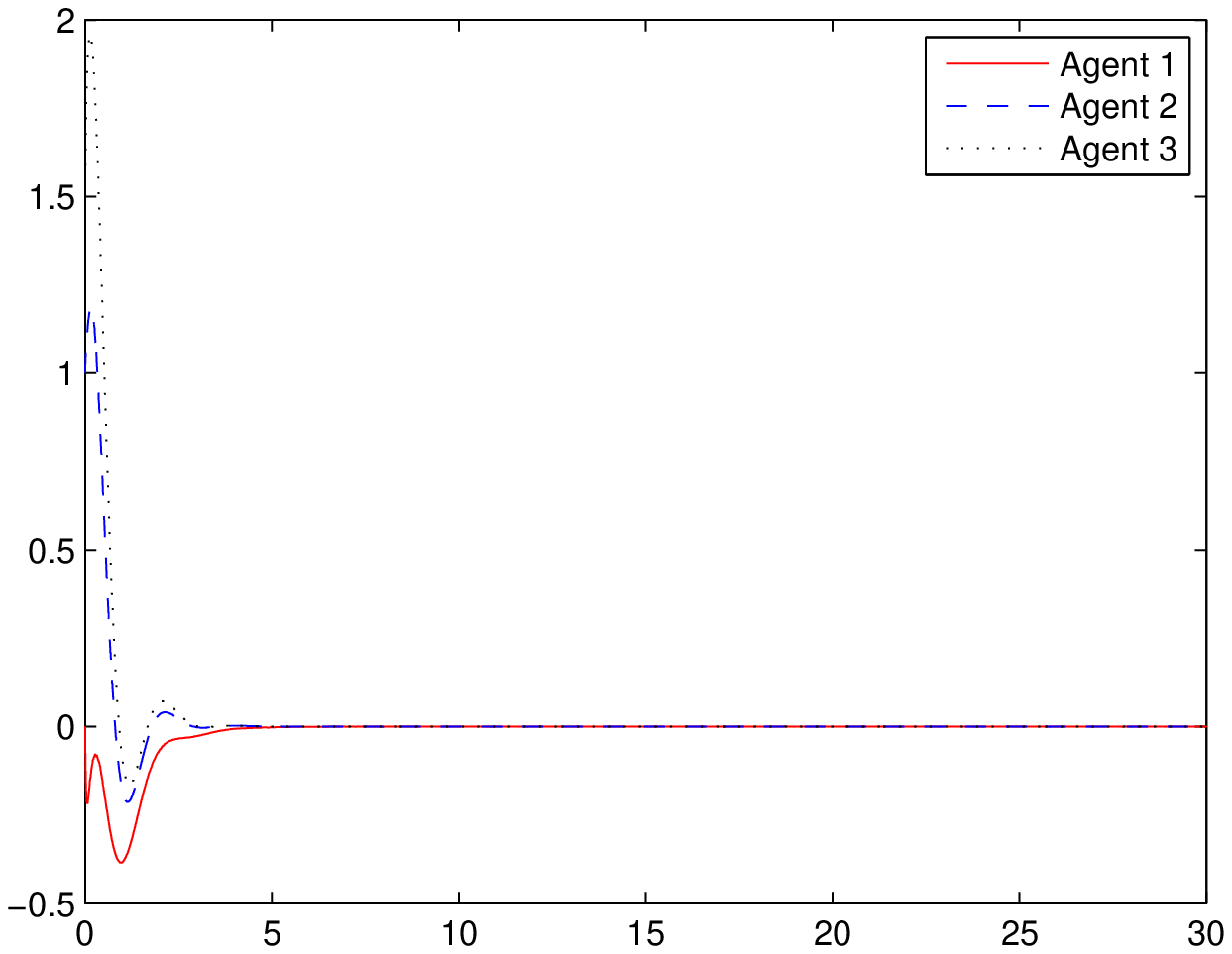}
\caption{Relative angles (the top figure) and relative velocities of the
  agents with
  respect to the leader, obtained using the algorithm based on
  Theorem~\ref{Theorem 1}.}
\label{Theorem1state1}
\end{figure}

\begin{figure}[htbp]
\centering
\includegraphics[width=.98\columnwidth]{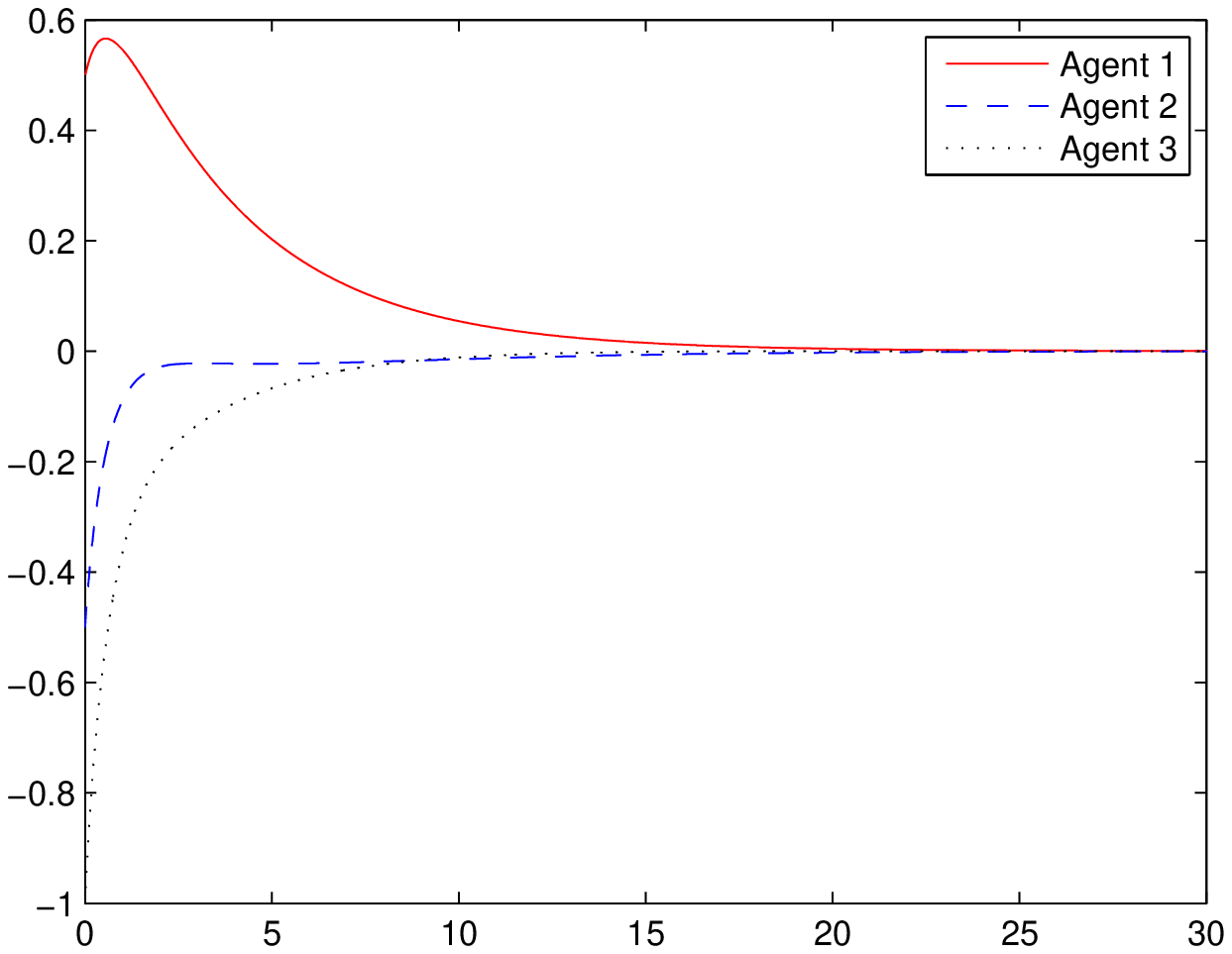}\\
\includegraphics[width=.98\columnwidth]{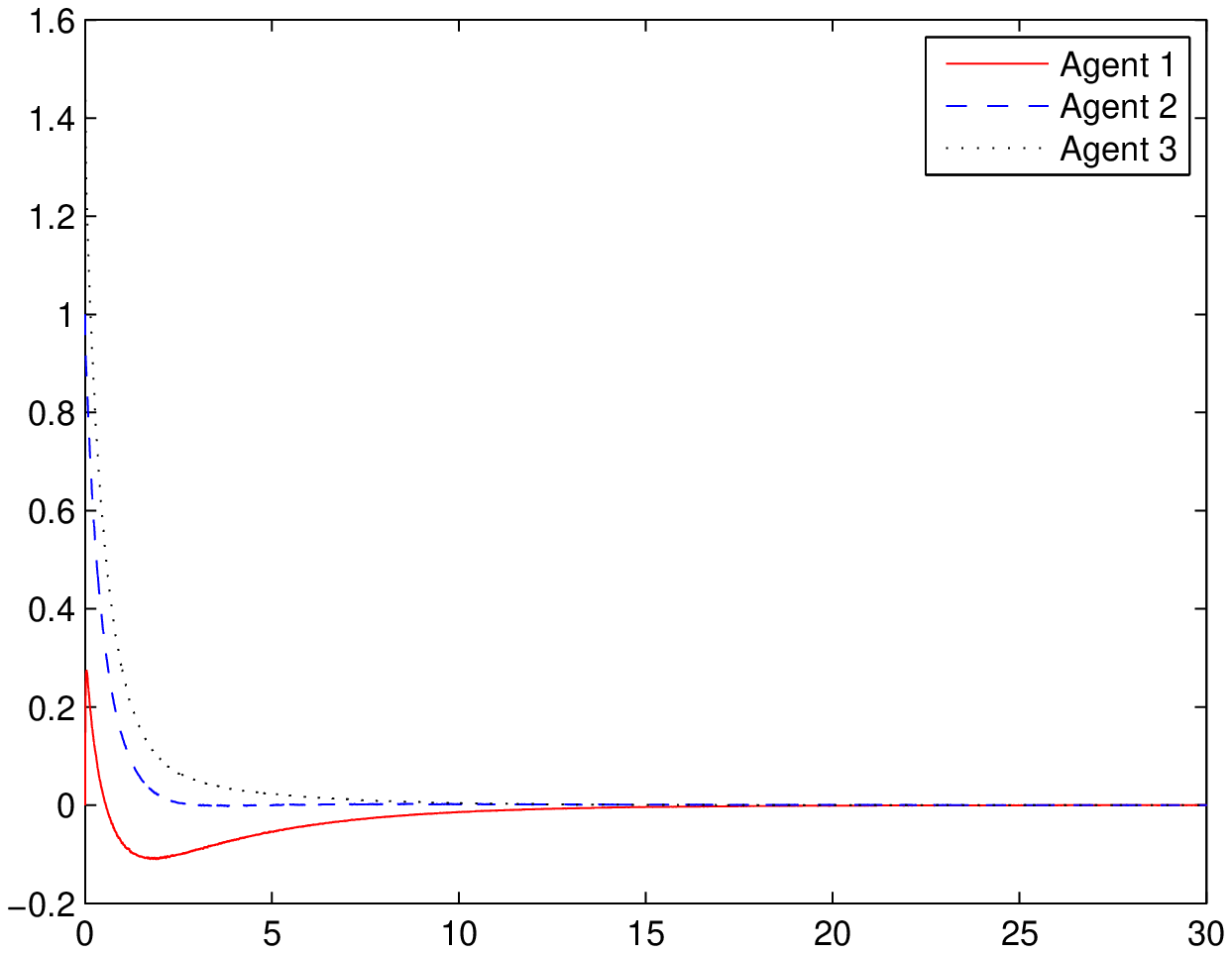}
\caption{Relative angles (the top figure) and relative velocities of the
  agents with respect to the leader, obtained using the algorithm based on
  Theorem~\ref{Theorem 2}.}
\label{Theorem2state1}
\end{figure}

\section{Conclusions}

The guaranteed consensus tracking performance problem for a coupled linear
system with undirected graph has been discussed in this paper. This
problem was transformed into a decentralized control problem for a system,
in which the interactions between subsystems satisfy IQCs. This has allowed
us to develop a procedure and sufficient conditions for the synthesis of
the tracking protocol for the original system.

According to the simulation results, the proposed computational algorithm
based on Theorem 1, which solves $N$ coupled LMIs, guarantees a better
performance. However, the proposed computational algorithm based on Theorem
2 only requires decoupled LMIs to be feasible, which enables the
controller gain to be computed in a decentralized manner.

% that's all folks
\end{document}